\documentclass[10pt,a4paper]{article}

\usepackage[margin=0.75in]{geometry}
\usepackage{authblk}
\usepackage{cite}
\usepackage{url}
\usepackage{tipa}
\usepackage[dvips]{graphicx}
\usepackage{amssymb}
\usepackage{color}
\usepackage{amsbsy}
\usepackage{textcomp}  
\usepackage{url}
\usepackage{booktabs}
\usepackage{multirow}
\usepackage{marvosym}

\usepackage[cmex10]{amsmath}

\newcommand{\bs}[1]{\boldsymbol{#1}}
\newcommand{\Rr}[1]{\mathfrak{Re} \left \{#1 \right\}}
\newcommand{\Ii}[1]{\mathfrak{Im} \left \{#1 \right\}}
\newcommand{\m}[1]{\mathcal{#1}}
\newcommand{\SNR}{\text{SNR}}

\newcommand{\bb}[1]{\bar{\textbf{#1}}}

\newcommand{\eor}{\hfill $\bigtriangledown$ \\}
\newcommand{\eot}{\hfill $\blacksquare$ \\}

\newcommand{\g}{\text{\textscg}}

\newtheorem{theorem}{Theorem}[section]
\newtheorem{lemma}[theorem]{Lemma}

\newtheorem{remark}[theorem]{Remark}

\newenvironment{proof}[1][Proof]{\begin{trivlist}
\item[\hskip \labelsep {\bfseries #1}]}{\end{trivlist}}

\newcommand{\qed}{\nobreak \ifvmode \relax \else
      \ifdim\lastskip<1.5em \hskip-\lastskip
      \hskip1.5em plus0em minus0.5em \fi \nobreak
      \vrule height0.75em width0.5em depth0.25em\fi}

\newenvironment{myindentpar}[1]%
	{\begin{list}{}%
	{\setlength{\leftmargin}{#1}}%
	\item[]%
    			}
	{\end{list}}

\title{\LARGE \textbf{A Bayesian approach to sparse channel estimation in OFDM systems}}

\author[$^\dagger$]{\Large Rodrigo Carvajal}
\author[$^\ddagger$]{\Large Boris I. Godoy}
\author[$^\ddagger$]{\Large Juan C. Ag\"uero}
\affil[$^\dagger$]{\normalsize Electronics Engineering Department, Universidad T\'{e}cnica Federico Santa Mar\'ia, Chile. \Letter~{rodrigo.carvajalg@usm.cl}}
\affil[$^\ddagger$]{\normalsize School of Electrical Engineering and Computer Science, The University of Newcastle, Australia. \Letter~\{boris.godoy,~juan.aguero\}@newcastle.edu.au}
\date{}

\begin{document}
\maketitle

\hrule \hrule
\begin{abstract}
\label{abstract}
\noindent In this work, we address the problem of estimating sparse
communication channels in OFDM systems in the presence of carrier frequency offset (CFO) and unknown noise variance. To this end, we consider a convex optimization problem, including a probability function,
accounting for the sparse nature of the communication channel. We use
the Expectation-Maximization (EM) algorithm to solve the corresponding
Maximum A Posteriori (MAP) estimation problem. We show that, by
concentrating the cost function in one variable, namely the CFO, the
channel estimate can be obtained in closed form within the EM
framework in the maximization step. We present an example where we
estimate the communication channel, the CFO, the symbol, the noise
variance, and the parameter defining the prior distribution of the
estimates. We compare the bit error rate performance of our
proposed MAP approach against Maximum Likelihood.

\end{abstract}
\vspace{2mm}
\hrule \hrule

\section{Introduction}
\label{sec:intro}

Sparse channel estimation is an important topic found in many
different applications (see e.g \cite{ref:Cotter2002,
  ref:Stojanovic2009, ref:Berger, ref:Taubock, ref:Kim} and the references therein). In fact, in many
real-world channels of practical interest, e.g. underwater acoustic
channels \cite{ref:Kilfoyle}, digital television channels
\cite{ref:ATSC}, and residential ultrawideband channels
\cite{ref:Molisch}, the associated impulse response tends to be
sparse. To obtain an accurate channel impulse response is crucial
since it is used in the decoding stage. Sparsity helps one can obtain
better channel estimates. In addition, the
most common technique for promoting sparsity is by an $\ell_1-$norm
regularization, commonly termed as Lasso \cite{ref:Tibshirani1996}. However, sparsity can be promoted in different ways. For example, in \cite{ref:Larsson07}, sparsity is promoted by generating a pool of possible models, and then performing model selection. 

A special characteristic of OFDM systems is its sensitivity to
frequency synchronizations errors (see e.g. \cite{ref:Carvajal2013}),
which is produced (among other causes) by carrier frequency offset
(CFO). This adds an extra difficulty to the channel estimation
problem, since the CFO must be estimated as well as other channel parameters. To estimate the
CFO, Maximum Likelihood (ML) estimation has been successfully utilized (see
e.g. \cite{ref:Moose, ref:Mo, ref:Carvajal2013}). 

In this work, we combine the following problems: (i) estimation of a
sparse channel impulse response (CIR) in OFDM systems, (ii) estimation
of CFO, (iii) estimation of the noise variance, (iv) estimation of the transmitted symbol, and (v) estimation of the
(hyper) parameter defining the prior probability density function
(pdf) of the sparse channel. The estimation problem is
solved by utilizing a generalization of the EM algorithm (see
e.g. \cite{ref:Mo, ref:Carvajal2013, ref:Carvajal12, ref:Godoy13} and the references
therein) for MAP estimation, based on the $\ell_1-$norm of the CIR. In particular, the same methodology has been applied in \cite{ref:Godoy13} for the identificaton of a sparse finite impulse response filter with quantized data. Our work generalizes previous work on
joint CFO and CIR estimation, see \cite{ref:Mo} and the generalization \cite{ref:Carvajal12}. 

The problem of estimating a sparse channel and the transmitted symbol has been previously addresses in the literature \cite{ref:Schniter}. In \cite{ref:Schniter}, it is also considered bit interleaved coded modulation (BICM) in OFDM systems. The approach in \cite{ref:Schniter} corresponds to the utilization of the \textit{generalized approximate message passing} (GAMP) algorithm \cite{ref:Rangan2}, which allows for solving the BICM problem. GAMP corresponds to a generalization of the \textit{approximate message passing} (AMP) algorithm \cite{ref:Donoho2009}, although it does not allow for unknown parameters other than the channel. The AMP and GAMP algorithms are based on belief propagation \cite{ref:Rangan2, ref:Donoho2009}. When the system is linear (with respect to the channel response), GAMP and AMP are the same algorithm \cite{ref:Rangan2}. In addition, for sparsity problems, the AMP algorithm corresponds to an efficient implementation of the Lasso estimator, see\cite{ref:Rangan2} and the references therein. Hence, under the same setup, the MAP-EM algorithm we propose and the GAMP algorithm utilized in \cite{ref:Schniter} yield the same results.


\section{OFDM System Model}
\label{section:ofdm_sysid}

We consider the following OFDM system model (see e.g. \cite{ref:Carvajal2013, ref:Carvajal12} and the references therein), depiected in Fig. \ref{fig:signal} \cite{ref:Carvajal2013}:
\begin{myindentpar}{0cm}
  \textbullet \hspace{1mm}The channel is modelled as a finite impulse   response (FIR) filter $\textbf{h} = [h_0 \,\, h_1 \,\, \ldots \,\, h_{L-1}]^T \in \mathbb{C}^{L}$ with $L$ taps.

  \textbullet \hspace{1mm}CFO is modelled by $\textbf{C}_{\varepsilon}
  = \exp \{j\mathrm{diag} \left( \frac{2\pi\varepsilon k}{ N_C}\right
  )\}$, with $k = 0,1,\ldots, N_C-1$; $\varepsilon$ is the
  \textit{normalized} frequency offset ($\vert \varepsilon \vert \leq
  1/2$, $N_C$ is the number of subcarriers).

  \textbullet \hspace{1mm} The cyclic prefix (CP) is removed at the
  receiver. Thus, the received signal is given by\vspace{-2mm}
\begin{equation}\label{eq:ofdm_model}
\textbf{r} = \textbf{C}_{\varepsilon}\tilde{\textbf{H}} {\sf P} \textbf{x} + \bs{\eta},
\vspace{-1mm}
\end{equation}
where the channel matrix $\tilde{\textbf{H}}$ is an $(N_C \times N_C)$
circulant matrix whose first column is given by $ [h_0, \,\, h_1, \,\,
\cdots \,\, $ $ h_{L-1},\,\,0 \,\, \dots \,\, 0]^T $, \textbf{x} is the
transmitted signal (after the inverse discrete Fourier transform), $\sf P$ is a permutation matrix that shuffles the transmitted symbol samples in any desired fashion \cite{ref:Carvajal2013}, $\bs{\eta}\sim {\mathcal N}(0,\sigma^2 {\bf I}_{N_C} ) $, and ${\bf I}_{N_C}$ is the identity matrix of dimension $N_C$. Notice that the
time-domain representation of the multicarrier signals in
\eqref{eq:ofdm_model} resembles a single-carrier system. However, the
main difference corresponds to the utilization of the cyclic prefix,
which yields a circulant channel matrix at the receiver after the
cyclic prefix removal.

\end{myindentpar}

\begin{figure}[t]
	\centering
	\scalebox{0.8}{\includegraphics[scale = 0.27, angle = 0]{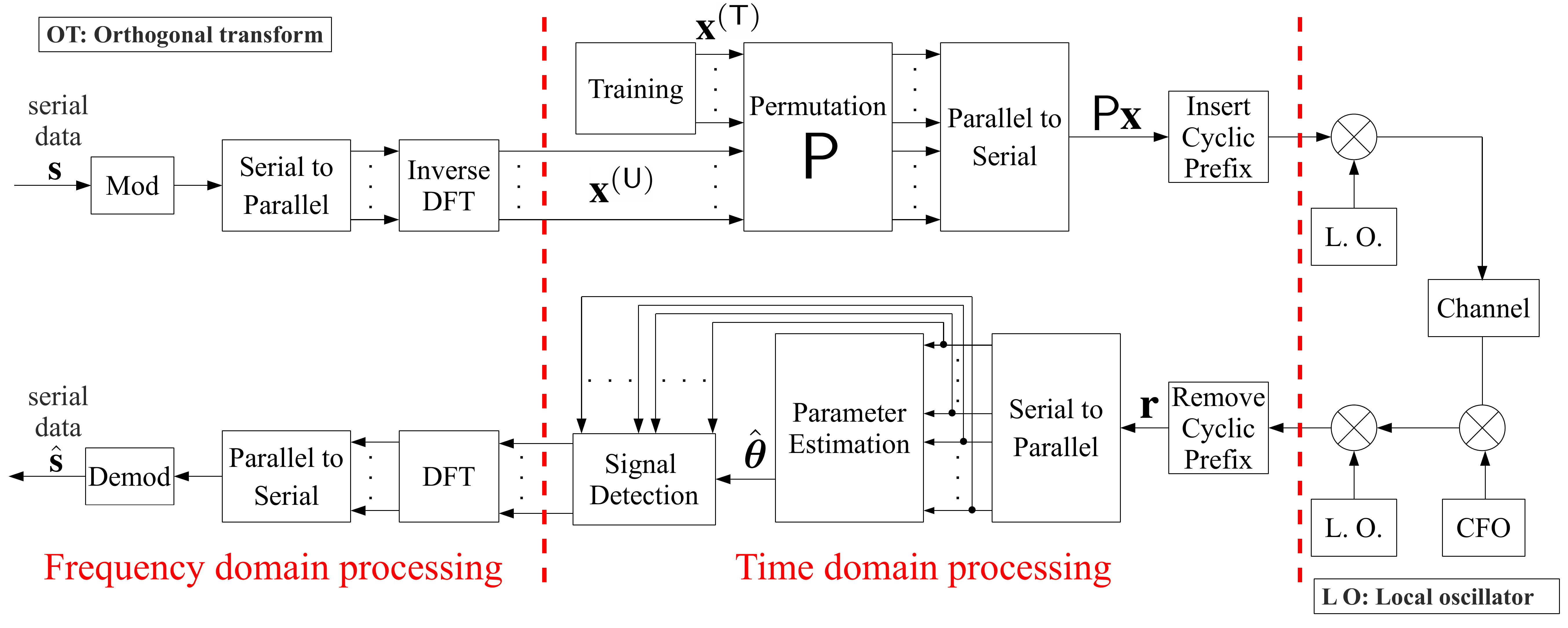}}
	\centering
	\vspace{-2mm}
	\caption{Block diagram of the general multicarrier system considered in this paper. LO: Local oscillator.}
	\vspace{-2mm}
	\label{fig:signal}
\end{figure}

The transmitted signal is assumed to have a deterministic part
(comprising known training data) and a stochastic part (comprising the
unknown data). Thus, the transmitted signal corresponds, after the
application of the IDFT, to the time domain multiplexing of a training
sequence and data coming from the data terminal equipment. We also
need to express the transmitted signal in terms of the known
(training) component, $\textbf{x}^{(\mathsf{T})}$, and the unknown
component, $\textbf{x}^{(\mathsf{U})}$. Thus, the real representation
of the transmitted signal $\textbf{x}$ is given by \vspace{-1mm}
\begin{equation}
\bar{\textbf{x}} =  [{\textbf{x}_{\text{\textscr}}^{\mathsf{(T)}}} ^T \; {\textbf{x}_{\text{\textscr}}^{\mathsf{(U)}}} ^T \;{\textbf{x}_{\text{\textsci}}^{\mathsf{(T)}}} ^T \;{\textbf{x}_{\text{\textsci}}^{\mathsf{(U)}}} ^T  ]^T \, \in \mathbb{R}^{2N_C},
\vspace{-1mm}
\end{equation}
where $(\cdot)_{\text{\textscr}}$, $(\cdot)_{\text{\textsci}}$,
$(\cdot)^{\mathsf{(T)}}$ and $(\cdot)^{\mathsf{(U)}}$ represent the
real part, imaginary part, training part, and unknown part,
respectively.

For estimation purposes, it is possible to express the  model in \eqref{eq:ofdm_model} as a real-valued state-space model with sample index $k$:
\begin{equation}
\bar{\textbf{y}}_k  = \left[	
\begin{array}{cc}	
\textbf{a}_k  &  -\textbf{b}_k\\
\textbf{b}_k  &  \textbf{a}_k
\end{array}	\right] \bb{x}+ \bar{\bs{\eta}}_k  = \bar{\textbf{M}}_k  \bb{x}+ \bar{\bs{\eta}}_k ,
\label{eq:sss}
\end{equation}
where $\bar{\textbf{y}}_k = [\Rr{r_k} \, \Ii{r_k}]^T$, $\bar{\bs{\eta}}_k = [\Rr{\eta_k} \, \Ii{\eta_k}]^T$,  $k = 0,1,...,N_C-1$
is the time sample index of the OFDM symbol, $\Rr{\cdot}$ and
$\Ii{\cdot}$ denote the real and imaginary parts, respectively,
\begin{align}
\textbf{a}_k = &  (\cos \psi_k ) \textbf{q}_{k+1}^T \Rr{\tilde{\textbf{H}}}
\sf{P} - (\sin \psi_k) \textbf{q}_{k+1}^T \Ii{\tilde{\textbf{H}}} \sf{P},
\\
\textbf{b}_k = & (\sin \psi_k) \textbf{q}_{k+1}^T \Rr{\tilde{\textbf{H}}}
\sf{P} + (\cos \psi_k) 
\textbf{e}_{k+1}^T \Ii{\tilde{\textbf{H}}} \sf{P},\\
\psi_k = & \frac{2\pi k \varepsilon}{N_C},
\end{align}
and $\textbf{q}_k$ is the $k$th column of the identity matrix. This
state-space representation is equivalent to \eqref{eq:ofdm_model}, but
it is more convenient for the identification approach used in this
work. In addition, and as it will be shown in Section \ref{section:Q_ML}, the estimation procedure is based upon expressions in the form of $E[\bb{x}^{(\sf U)}|\bb{y}]$ and $E[\bb{x}^{(\sf U)}{\textbf{x}^{\mathsf{(U)}}} ^T|\bb{y}]$, amongst other quantities. The attainment of these two expectations can be achieved, for instance, by applying Bayes' rule for the \textit{posterior} pdf
\begin{equation}
p(\textbf{x}^{\mathsf{(U)}}|\bb{y}) = \frac{p(\bb{y}|\textbf{x}^{\mathsf{(U)}}) p(\textbf{x}^{\mathsf{(U)}})}{p(\bb{y})},
\end{equation}
for any given \textit{prior} pdf $p(\textbf{x}^{\mathsf{(U)}})$.
\begin{remark}
It is possible to extend the state-space model in \eqref{eq:sss} by including a constant state vector that corresponds to the whole unknown transmitted signal (see e.g \cite[Chap. 9]{ref:Soderstrom}). That is,\textup{
\begin{align}
\bs{\chi}_{k+1} & =  \,\,\bs{\chi}_{k} = \bb{x}^{(\mathsf{U})}, \nonumber \\
\bar{\textbf{y}}_k & = \bar{\textbf{M}}_k  \bb{x}+ \bar{\bs{\eta}}_k. 
\label{eq:sss_ext} \vspace{-1mm}
\end{align} }
Notice that the subindex $k$ for $\bs{\chi}$ in \eqref{eq:sss_ext} indicates that $\bs{\chi}$ remains unchanged for every sample index $k = 0,1,...,N_C-1$. This extension allows for the utilization of \textit{filtering techniques} for the attainment of \textup{$p(\textbf{x}^{\mathsf{(U)}}|\bb{y})$} (and consequently \textup{$E[\bb{x}^{(\sf U)}|\bb{y}]$} and \textup{$E[\bb{x}^{(\sf U)}{\textbf{x}^{\mathsf{(U)}}} ^T|\bb{y}]$}).
\end{remark}


\begin{remark} 
  We consider a general state-space model that can be utilized for
  proper and improper signals\footnote{This representation also
    extends to proper and improper CIR and additive noise.}
  \cite{ref:Miller}. In this sense, our approach can be applied to all
  common modulation schemes, such as binary phase shift keying (BPSK)
  and Gaussian minimum shift keying (GMSK), which are improper (see
  e.g. \cite{ref:Buzzi}). \eor
\vspace{-2mm}
\end{remark}

Regarding the received signal, the conditional pdf of $\textbf{y}=
[\bar{\textbf{y}}^T_{0},\dots ,\bar{\textbf{y}}^T_{N_C-1}]^T$ is given by $p(\textbf{y}\,\vert \, \bs{\chi}, {\bs{\theta}}) \sim \m{N}(\textbf{M}\bb{x},\boldsymbol{\Sigma}_{y})$
where the vector of parameters $\bs{\theta} = (\bb{h} , \varepsilon,
\sigma )$, $\bb{h} = [\Rr{\textbf{h}} \, \Ii{\textbf{h}} ]^T$, and
\vspace{-2mm}
\begin{equation*}
\textbf{M} = \begin{bmatrix}
\bar{\textbf{M}}_0^T & \cdots &
\bar{\textbf{M}}_{N_C-1}^T \end{bmatrix}^T, \quad \bs{\Sigma}_y =
\mathbb{E}[\hat{\bs{\eta}} \hat{\bs{\eta}}^T], \vspace{-2mm} 
\end{equation*}
with $\hat{\bs{\eta}} = [\bar{\bs{\eta}}_0^T \,
\cdots\,\bar{\bs{\eta}}_{N_C-1}^T]$. For proper additive white noise with
variance $\sigma^2$, $\boldsymbol{\Sigma}_{y} = 0.5\sigma^2 {\bf
  I}_{2N_C}$.
For the unknown part of the transmitted signal, the corresponding pdf is simply expressed as $p(\bb{x}^{(\sf{U})})$. This allows the transmitted signal to be generated by any modulation scheme, yielding a family of possible pdf's.
\begin{remark} 
  In this work, we consider a time-domain processing for the
  attainment of the estimates. This approach yields a special
  structure that closely resembles single-carrier systems
  \cite{ref:Wang2004, ref:Amleh2010}. However, to be consistent with
  previous works, we present our method in the OFDM framework. Because
  of the similarity of OFDM systems in time domain and single-carrier
  systems, the present algorithm can also be modified to cover the
  latter case.
\eor
\vspace{-2mm}
\end{remark}

\section{MAP estimation in OFDM systems}

To promote sparsity in the parameter $\bb{h}$, we include an $\ell_1$
regularization term in the form of a prior distribution, $p(\bb{h})$,
which, in turn, leads to a MAP estimation problem. In general, MAP
estimation allows for the inclusion of one or more terms that account
for statistical prior knowledge of the parameters
$\bs{\theta}$. However, here we are interested in utilizing prior
knowledge of the channel impulse response only. 
On the other hand, in a MAP estimation problem, a good estimate
$\hat{\sigma}$ is crucial. Thus, it is important to take into account
$\sigma$ in the definition of the problem (see e.g.
\cite{ref:stadler10}). If this is not done, the regularized
optimization problem may be non-convex and exhibit numerical
difficulties. To address this issue, we can express the prior distribution
for $\bb{h}$ as
\begin{align}\label{eq:p_beta}
p(\bb{h}) &= p(\bb{h}|\sigma,\varepsilon ) p(\sigma)p(\varepsilon),
\quad \text{with}\\ 
\label{eq:p_theta_sigma}
p(\bb{h} | \sigma) &= \left(\frac{1}{2\sigma\tau}\right)^{2L} \exp \left\lbrace- \frac{|
  \bb{h}|}{\tau \sigma}  \right\rbrace.
\end{align}

Since we assume no prior knowledge for the channel noise variance nor
the CFO, we choose non-informative marginal prior distributions, for
example, $p(\sigma) = p(\epsilon) = 1$ (see
\cite{ref:stadler10}). Then, the maximization problem becomes
\begin{equation}\label{eq:log_map1}
  \hat{\bs{\theta}}     = \mathrm{arg} \max_{\bs{\theta}} \,\, [\log\,p( \textbf{y} | \bs{\theta})
  + \log\,p(\bb{h}|\sigma ) ].
\end{equation}
where $p(\textbf{y}|\bs{\theta})$ is the \textit{likelihood} function.

To achieve convexity, we now can use the procedure suggested in
\cite{ref:stadler10}. That is, we introduce the following
reparametrization: $\textbf{\textscg} = \bb{h}/
\sigma$, $\rho = \sigma^{-1}$. We therefore define a new parameter to
be estimated:
$\bs{\gamma} = (\textbf{\textscg}, \varepsilon, \rho)$.
Using the new parametrization, and looking at the second term of the
right hand side of \eqref{eq:log_map1}, this term can expressed as a
function of $\textbf{\textscg}$ (or equivalently of the ``individual''
terms of $\textbf{\g}$, $\g_j$ is the $j$th element of
$\textbf{\textscg}$, $j = 1,2,...,2L$, as (see e.g. \cite{ref:Polson})
\vspace{-2mm}
\begin{equation}\label{eq:log_pen_sum}
  \log\,p(\textbf{\textscg}) = \sum_{j =
    1}^{2L}z\left(\frac{\g_j}{\tau }\right),\vspace{-2mm}
\end{equation}
where $z(\cdot)$ is a function specifying the log-prior.


\subsection{The EM algorithm and MAP estimation}

The EM algorithm is an iterative method that generates a succession of
estimates $\hat{\bs{\gamma}}^{(i)} = ( \hat{ \textbf{\textscg}}^{(i)},
\hat{\varepsilon}^{(i)}, {\hat{\rho} }^{(i)} )$, $i = 1,2, \ldots,$ of
the parameters $\bs{\gamma}$, which converges to a local maximum of
the log-likelihood function (see e.g. \cite{ref:Dempster}). The EM
algorithm consists of an iterative two-step procedure: i)
an expectation step (E-step), and ii) a maximization step (M-step). In
our case, we develop an augmented EM algorithm to solve the MAP
estimation problem presented in this work. The E-step consists of
computing the auxiliary function \vspace{-2mm}
\begin{equation}\label{eq:M_mod}
  \m{Q}(\bs{\gamma},\hat{\bs{\gamma}}^{(i)} ) = \m{Q}_\text{ML}(\bs{\gamma},\hat{\bs{\gamma}}^{(i)} ) + \m{Q}_{\text{prior}}(\textbf{\textscg},\hat{\textbf{\textscg}}^{(i)} ) ,\vspace{-2mm}
\end{equation}
where $\m{Q}_{\text{prior}}(\textbf{\textscg},\hat{\textbf{\textscg}}^{(i)}
)$ is the function corresponding to the a priori distribution.
The function  $\m{Q}_\text{ML}(\bs{\gamma},\hat{\bs{\gamma}}^{(i)})$ is the typical auxiliary function arising from the related ML estimation problem, given by \vspace{-2mm}
\begin{equation}\label{eq:Q}
  \m{Q}_\text{ML}(\bs{\gamma}, \bs{\gamma}^{(i)} ) := \mathbb{E}\left
    [\log p(\bb{x}^{(\sf U)},\textbf{y}\,\vert
    \,\bs{\gamma})]\,\vert \textbf{y},\hat{\bs{\gamma}}^{(i)}
  \right]. \vspace{-2mm}
\end{equation}

On the other hand, the M-step consists of maximizing the
auxiliary function
$\m{Q}(\bs{\gamma},\hat{\bs{\gamma}}^{(i)} )$, yielding 
\begin{equation}
  \hat{\bs{\gamma}}^{(i+1)} = \mathrm{arg} \max_{\boldsymbol{\gamma}}
  \mathcal{Q}(\bs{\gamma},\hat{\bs{\gamma}}^{(i)} ). 
\label{eq:M}
\end{equation}


\subsection{Evaluation of $\m{Q}_{\text{ML}}(\boldsymbol{\gamma},\hat{\boldsymbol{\gamma}}^{(i)})$ and its derivative}
\label{section:Q_ML}

The E-step of the EM algorithm given in \eqref{eq:Q} can be expressed
as \vspace{-3mm}
\begin{align}
\m{Q}_{\text{ML}}&(\boldsymbol{\gamma},\hat{\boldsymbol{\gamma}}^{(i)}
)  = \mathbb{E}\left[\log p(\bb{x}^{(\sf U)})\vert \textbf{y},
  \hat{\boldsymbol{\gamma}}^{(i)}\right]+ K_y \nonumber \\ \label{eq:EM2}
&  -\frac{1}{2}\mathbb{E}\left[(\textbf{y}-\textbf{M} \bb{x})^T
  {\boldsymbol \Sigma}_y^{-1}(\textbf{y}- \textbf{M} \bb{x}) \vert \textbf{y}, \hat{\boldsymbol{\gamma}}^{(i)}\right],\vspace{-5mm}
\end{align}
where $\textbf{M}$ is a (matrix) function of the
parameters $\bs{\gamma}$,  and $K_y = -N_C \log(2\pi) -0.5N_C \log
(0.5) + N_C \log(\rho^{2} )$. In addition, we define \vspace{-1mm}
\begin{align*}
  \bs{\m{M}}^T &= \begin{bmatrix}
    \bar{\bs{\m{M}}}_0^T &  \vdots & \bar{\bs{\m{M}}}_{N_C-1}^T
\end{bmatrix},\quad 
\bar{\bs{\m{M}}}_k = \begin{bmatrix}
\bb{a}_k & - \bb{b}_k\\
\bb{b}_k &  \bb{a}_k
\end{bmatrix}, \\
\bb{a}_k= & \,(\cos \psi_k ) \textbf{e}_{k+1}^T {\tilde{\textbf{X}}}_\text{\textscr}\m{I} - (\sin \psi_k) \textbf{e}_{k+1}^T{\tilde{\textbf{X}}}_\text{\textsci}\m{I},
\\
\bb{b}_k = & \,(\sin \psi_k) \textbf{e}_{k+1}^T{\tilde{\textbf{X}}}_\text{\textscr}\m{I} + (\cos \psi_k) \textbf{e}_{k+1}^T {\tilde{\textbf{X}}}_\text{\textsci}\m{I},
\end{align*}
with $\tilde{\textbf{X}}_{\text{\textscr}}$ and
$\tilde{\textbf{X}}_{\text{\textsci}}$ being the circulant matrices
generated by $\sf{P}\textbf{x}_{\text{\textscr}}$ and
$\sf{P}\textbf{x}_{\text{\textsci}}$, respectively, and $\m{I} = [
\textbf{I}_{L} \, \textbf{0}]^T$. Then, we can write ${\textbf M}
\bb{x} = \bs{\m{M}} \textbf{\textscg} \rho^{-1} $. Replacing this
equality in \eqref{eq:EM2}, and taking the derivative of
$\m{Q}_{\text{ML}}(\bs{\gamma},\hat{\bs{\gamma}}^{(i)})$ with respect
to $\textbf{\textscg}$, and $\rho$, we obtain:
\begin{align}
\label{eq:dQ_dh1}
\frac{\partial \m{Q}_{\text{ML}}}{\partial \textbf{\textscg}} & =
2\left( \rho \textbf{y}^T\mathbb{E}[\bs{\m{M}}\vert  \textbf{y},
\hat{\boldsymbol{\gamma}}^{(i)}] -   \mathbb{E}[ \bs{\m{M}}^T\bs{\m{M}} \vert
\textbf{y}, \hat{\boldsymbol{\gamma}}^{(i)} ] \textbf{\textscg} \right),\\
\frac{\partial \m{Q}_{\text{ML}}}{\partial \rho } & = \frac{2
  N_C}{\rho} -  2\left( \rho \textbf{y}^T \textbf{y} -  \textbf{y}^T
  \mathbb{E}\left[(\bs{\m{M}})\vert  \textbf{y},
    \hat{\boldsymbol{\gamma}}^{(i)}\right]
  \textbf{\textscg}\right). 
\label{eq:dQ_dsigma2}
\end{align}

\subsection{Evaluation of $\m{Q}_{\text{prior}}(\textbf{\textscg},\hat{\textbf{\textscg}}^{(i)} )$ and
  its derivative}
\label{subsec:Qprior}

We express $z(\cdot)$ as a variance-mean Gaussian mixture (VMGM)
\cite{ref:Polson}. When expressed in terms $\g_j,\, j = 1,2,...$, a
VMGM for the parameters is given by (see e.g \cite{ref:Polson,
  ref:Barndorff-Nielsen})\vspace{-1mm}
\begin{equation}
p(\textbf{\textscg}) = \prod_j \int_0^{\infty} p({\g_j}|{\lambda_j})p({\lambda_j})d{\lambda_j},\label{eq:mvgm} \vspace{-2mm}
\end{equation}
where $ {\g_j}|{\lambda_j} \sim \m{N}_{\g_j} (0 ,\lambda_j
\tau^2 )\,, \lambda_j \sim p(\lambda_j)$. In this sense, the random
variable $\lambda_j $ ($> 0$) can be considered as a \textit{hidden
  variable} in the EM algorithm. Hence, given \eqref{eq:mvgm}, the
auxiliary function
$\m{Q}_{\text{prior}}(\textbf{\g},\hat{\textbf{\g}}^{(i)})$, can be
expressed as\vspace{-2mm}
\begin{equation*}
  \hspace{-10mm} \m{Q}_{\text{prior}} (\textbf{\textscg},\hat{\textbf{\textscg}}^{(i)}) =
  \sum_{j = 1}^{2L}\int \log [
  p(\g_j,\lambda_j)] p(\lambda_j|\hat{\g}_j^{(i)})d(\lambda_j)
   \vspace{-3mm} 
\end{equation*}
\begin{equation}
  = \sum_{j = 1}^{2L}\int \left ( \log
  [p(\g_j|\lambda_j)]+\log
  [p(\lambda_j)] \right )p(\lambda_j|\hat{\g}_j^{(i)})d(\lambda_j),\vspace{-2mm}
  \label{eq:Q_prior1}
\end{equation}
since  $\log p(\textbf{\g},\bs{\lambda}) = \sum_{j = 1}^{2L} \log$  $ p(\g_j,\lambda_j)$. 
\begin{lemma}
In the case that \textup{$\m{Q}_{\text{prior}}(\textbf{\g},\hat{\textbf{\g}}^{(i)})$} is given by \eqref{eq:Q_prior1}, then its derivative is given by \textup{
\begin{align}
\! \frac{\partial
  \m{Q}_{\text{prior}}(\textbf{\g},\hat{\textbf{\g}}^{(i)})}{\partial
  \g_j}  
  & = \left[  - \frac{\g_j }{\tau^2 }
    \mathbb{E}_{\lambda_j|\hat{\g}_j^{(i)}}\{\lambda_j^{-1}\} \right],
\label{eq:dQprior_dtheta}
\end{align}}
where \textup{$\mathbb{E}_{\lambda_j|\hat{\g}_j^{(i)}}\{\lambda_j^{-1} \}$} is the
expectation obtained from \vspace{-2mm}\textup{
\begin{equation}\label{eq:dlog_pen_eq}
- \frac{\hat{\g}^{(i)}_j }{\tau^2
   }\mathbb{E}_{\lambda_j|\hat{\g}^{(i)}_j}\{\lambda_j^{-1}\}  = \dot{z}\left(\frac{\hat{\g}^{(i)}_j}{\tau  }\right).
\end{equation}}
\end{lemma}
\begin{proof}
See \cite{ref:Carvajal12}. \eot
\end{proof}

From the M-step, at the $i$th iteration, an estimate
($\hat{\g}_j^{(i)}$) of $\g_j$ is obtained. This estimate is,
then, inserted into \eqref{eq:dlog_pen_eq}, in order to obtain an
estimate of $\mathbb{E}_{\lambda_j| \hat{\g}^{(i)}_j}
\{\lambda_j^{-1} \}$, which in turn is utilized in the maximization of
$\m{Q}_\text{prior}$. Once the new estimate $\hat{\g}_j^{(i+1)}$
has been obtained, it is inserted into \eqref{eq:dlog_pen_eq} and the
iteration continues until convergence has been reached.

In our particular case, we only want to promote sparsity in
$\textbf{\g}$ (consequently in the CIR $\bb{h}$). Thus, our chosen
penalty function is $z(\g_j / \tau ) =
|\g_j/ \tau|$. Using \eqref{eq:dlog_pen_eq}, we have that
$\mathbb{E}_{\lambda_j|\hat{\g}_j}[\lambda_j^{-1}] =
-\tau \text{sign}(\hat{\g}^{(i)}_j)/
\hat{\g}^{(i)}_j$. Using this value for
$\mathbb{E}_{\lambda_j| \hat{\g}_j}\{\lambda_j^{-1}\}$,
and calculating ${\partial \m{Q}_{\text{prior}} /\partial 
\textbf{\g}} $ we have that 
\begin{align}
 \frac{\partial \m{Q}_{\text{prior}}}{\partial
   \textbf{\textscg} } = -\frac{1}{\tau^2}  
 \textbf{E} \textbf{\textscg},
 \label{eq:dQprior_dhbar} 
\vspace{-3mm}
\end{align} 
where $\textbf{E} = \text{diag}\left(
  \mathbb{E}_{\lambda_1|\hat{\g}_1^{(i)}}\{ \lambda_1^{-1} \},\,
  ... ,\,
  \mathbb{E}_{\lambda_{2L}|\hat{\g}_{2L}^{(i)}} \{
  \lambda_{2L}^{-1} \}  \right)$. \vspace{-2mm}

\subsection{Combination of   $\m{Q}_{\text{ML}}(\bs{\gamma},\hat{\bs{\gamma}}^{(i)})$ and  $\m{Q}_{\text{prior}}(\textbf{\g},\hat{\textbf{\g}}^{(i)})$}

We are building our strategy on an underlying ML estimation algorithm. Thus, we assume
$\m{Q}_{\text{ML}}$ and $\partial \m{Q}_\text{ML}/\partial \textbf{\g}$ known. The strategy is then to derive the augmented E-step considering
both $\m{Q}_\text{ML}$ and $\m{Q}_{\text{prior}}$ with respect to
$\textbf{\textscg}$, that is,
\begin{equation}
\frac{\partial \m{Q}}{\partial \textbf{\textscg} } = \frac{\partial
  \m{Q}_{\text{ML}}}{\partial \textbf{\textscg} }  + \frac{\partial
  \m{Q}_{\text{prior}}}{\partial \textbf{\textscg} }.
\label{eq:dQ_dh}
\end{equation}
Using \eqref{eq:dQ_dh1}, \eqref{eq:dQprior_dtheta},
\eqref{eq:dQprior_dhbar}, and \eqref{eq:dQ_dh}, and expressing
$\textbf{\textscg}$ as a function of $\varepsilon$, we have
\vspace{-1mm}
\begin{equation}\label{eq:h_bar}
\textbf{\textscg} =  \left[\mathbb{E}[\bs{\m{M}}^T \bs{\m{M}}\vert \textbf{y},
  \hat{\boldsymbol{\gamma}}^{(i)} ]
+ \frac{1}{2\tau^2} \textbf{E}  \right]^{-1} \! \! \mathbb{E}[\bs{\m{M}}\vert  \textbf{y},
\hat{\boldsymbol{\gamma}}^{(i)}]^T \rho \textbf{y}.
\vspace{-1mm}
\end{equation}
Replacing the expression for $\textbf{\textscg}$ in \eqref{eq:M_mod}, we can
optimize $\m{Q}$ in \eqref{eq:M_mod} with respect to the parameter
$\varepsilon$. Thus, the parameter $\textbf{\textscg}$ (consequently
$\bar{\textbf{h}}$) is obtained by replacing the result of the
optimization for $\varepsilon$ in \eqref{eq:h_bar}. One advantage of
our method is that it allows the concentration of the cost in one
variable, namely CFO. In addition, we obtain closed form expressions
for the optimization of the regularized communication channel, namely
CIR, which, in general, is not possible with other methods when
applying $\ell_1$-norm regularization.

\subsection{Estimation of $\tau$}
\label{subsec:tau_est}

So far, the proposed algorithm for sparse channel estimation relies upon knowledge of $\tau$ (or at least a good estimate of it). Knowledge of this variable is important for accurate
estimates of $\bb{h}$. However, having \emph{a priori} knowledge of this
parameter is not always possible. For example, in an urban cellular
network, the channel can exhibit different behaviours depending on the
location, presenting the possibility of having different values of
$\tau$ (at each one of the locations). Thus, we seek an estimate of $\tau$.
\begin{lemma}
Using \eqref{eq:p_theta_sigma}, the Empirical-Bayes (EB) estimate \textup{$\hat{\tau}_{\text{EB}}$} is given by \vspace{-3mm}\textup{
\begin{equation}\label{eq:tau_est}
\hat{\tau}_{\text{EB}} = \frac{ \mathbb{E}_{|\bb{h}|/\sigma  | \bb{y},
    \hat{\tau}^{(i)} }\left[ \frac{|\bb{h}| }{\sigma}\right]}  {2L}.
    \vspace{-1mm}
\end{equation}}
\end{lemma}
\begin{proof}
  We define an auxiliary function $Q(\tau,\tau^{(i)}) =
  2L\log(\tau^{-1}/2\sigma) - \tau^{-1} \mathbb{E}_{|\bb{h}|/\sigma |
    \bb{y}, \hat{\tau}^{(i)} }\left[ \frac{|\bb{h}| }{\sigma}\right]$,
  take derivative with respect to $\tau^{-1}$, and then set the result equal to zero. \eot
\end{proof}
In general, the computation of $\mathbb{E}_{|\bb{h}|/\sigma | \bb{y},
  \hat{\tau}^{(i)} }\left[ \frac{|\bb{h}| }{\sigma}\right]$ is
computational expensive, requiring, in addition, many observations
$\bb{y}$. To avoid this problem, we approximate
$\mathbb{E}_{|\bb{h}|/\sigma | \bb{y}, \hat{\tau}^{(i)} }\left[
  \frac{|\bb{h}| }{\sigma}\right]\approx
| \hat{\bb{h}}_{\text{ML}}| /\hat{\sigma}_{\text{ML}}$, where
$\hat{\bb{h}}_{\text{ML}}$ and $\hat{\sigma}_{\text{ML}}$ are the ML
estimates using no regularization term, and where $\bar{\textbf{x}}$ is 
completely known (100\% training).

The solution to the regularized estimation problem (considering the
reparametrization) presented in this work can be summarized in the
following steps:
\begin{itemize}
\item[(i)] with 100\% training, and no regularization term, calculate
 $\hat{\tau} \approx \hat{\bb{h}}_{\text{ML}}/2L\hat{\sigma}_{\text{ML}}$,
\item[(ii)] $\hat{{\bs{\theta}}}^{(i)} = ({\hat{\bb{h}}}^{(i)} ,
  \hat{\varepsilon}^{(i)}, ({\hat{\sigma}^2})^{(i)} )$, and form the new
  variables: $\textbf{\textscg}^{(i)} = \hat{\bb{h}}^{(i)}/\hat{\sigma}^{(i)}$,
  and $\rho^{(i)} = 1/\hat{\sigma}^{(i)}$,
\item[(iii)] with a fixed $(\hat{\sigma}^2)^{(i)}$ from (ii), optimize
  for $\varepsilon$ after replacing \eqref{eq:h_bar} in \eqref{eq:M_mod},
\item[(iv)] with the estimate ($\hat{\textbf{\textscg}}^{(i+1)},
  \hat{\varepsilon}^{(i+1)}$) (consequently $\hat{\bb{h}}^{(i+1)}$) obtained
  in (iii), find $\hat{\sigma}^{(i+1)}$ from making zero the
  right-hand side of \eqref{eq:dQ_dsigma2}, and solve a quadratic
  equation, 
\item[(v)] go back to (ii) until convergence.
\end{itemize}

\section{Numerical Example}

\begin{table}[t]
\caption{N-MSE performance for MAP and ML estimation. Unknown channel noise variance.}   
\vspace{5mm}
\setlength{\tabcolsep}{2.7pt}
\centering                          
\footnotesize{
\begin{tabular}{c c c c}             
\toprule
Approach & Training & Number of iterations & MSE   \\ [0.5ex]   
\cmidrule(r){1-4}                         
\multicolumn{1}{c}{\multirow{2}{*}{ML (SNR = 5[db])}}  & $100\%$ & $100$ & $2.95\times 10^{-2}$\\
&  $62.5\%$ & $300$ & $9.24\times 10^{-2}$\\ 
\multicolumn{1}{c}{\multirow{2}{*}{MAP (SNR = 5[db])}}  & $100\%$ & $100$ & $2.10\times 10^{-2}$\\
&  $62.5\%$ & $300$ & $4.69\times 10^{-2}$\\ \hline
\multicolumn{1}{c}{\multirow{2}{*}{ML (SNR = 10[db])}}  & $100\%$ & $100$ & $9.30\times 10^{-3}$\\
&  $62.5\%$ & $300$ & $2.67\times 10^{-2}$\\ 
\multicolumn{1}{c}{\multirow{2}{*}{MAP (SNR = 10[db])}}  & $100\%$ & $100$ & $5.80\times 10^{-3}$\\
&  $62.5\%$ & $300$ & $1.27\times 10^{-2}$\\
\bottomrule                             
\end{tabular}}
\label{table:variance_u}          
\end{table}

\begin{figure}[t]
	\hspace{-2mm}
\begin{center}
		\includegraphics[width=0.6\columnwidth]{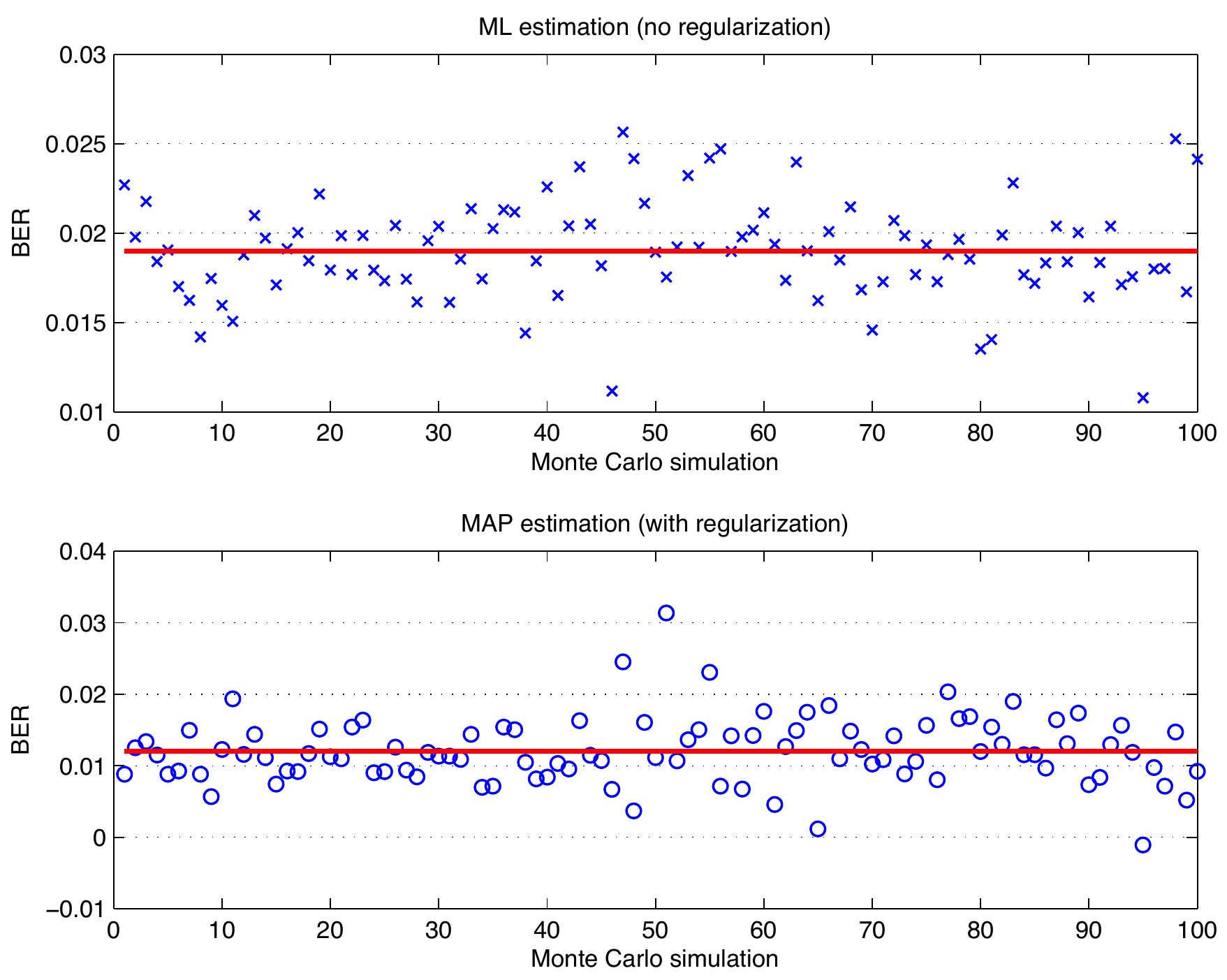}
\end{center}
	\vspace{-5mm}
        \caption{ ML estimation (upper plot),
          MAP estimation (lower plot). Continuous
          line: average value. $\text{SNR} = 10[\text{dB}]$.}
    \label{fig:ber}
\end{figure}

In this section, we present a numerical example using our approach for
an OFDM system with CFO. We assume that the unknown part of the time-domain transmitted signal is approximately Gaussian distributed (a consequence of the Central Limit Theorem). Thus, $p(\bb{x}^{(\sf{U})}) \sim \m{N} \left( \textbf{0}, \bs{\Sigma}_{\bar{x}^{({\sf U})}} \right)$, where $\bs{\Sigma}_{\bar{x}^{({\sf U})}} = \begin{bmatrix}
\bs{\Sigma}_{\textbf{x}_{\text{\textscr}}^{({\sf U})}} &
\bs{\Sigma}_{\textbf{x}_{\text{\textscr}}^{({\sf U})}
  \textbf{x}_{\text{\textsci}}^{({\sf U})}}\\
\bs{\Sigma}_{\textbf{x}_{\text{\textsci}}^{({\sf U})}
  \textbf{x}_{\text{\textscr}}^{({\sf U})}} &
\bs{\Sigma}_{\textbf{x}_{\text{\textsci}}^{({\sf U})}}
\end{bmatrix}$,
$\bs{\Sigma}_{\textbf{x}_{\text{\textscr}}^{({\sf U})}} = E[\textbf{x}_{\text{\textscr}}^{({\sf U})} {\textbf{x}_{\text{\textscr}}^{({\sf U})}}^T] $, $\bs{\Sigma}_{\textbf{x}_{\text{\textsci}}^{({\sf U})}} = E[\textbf{x}_{\text{\textsci}}^{({\sf U})} {\textbf{x}_{\text{\textsci}}^{({\sf U})}}^T]$, and $\bs{\Sigma}_{\textbf{x}_{\text{\textsci}}^{({\sf U})}
  \textbf{x}_{\text{\textscr}}^{({\sf U})}} = E[ \textbf{x}_{\text{\textsci}}^{({\sf U})} {\textbf{x}_{\text{\textscr}}^{({\sf U})}}^T ]$ (known).

The expectations on the right hand side of \eqref{eq:dQ_dh1} and
\eqref{eq:dQ_dsigma2} can be readily calculated by applying Kalman
filtering to the model in \eqref{eq:sss_ext}.  In addition, we consider
that channel noise variance is unknown. We consider the following
set-up: (i) $N_C = 64$, (ii) a sparse channel impulse response of
length $20$, with $14$ taps equal to zero, (iii) the transmitted
signal is Gaussian distributed, (iv) the signal to noise ratio is
$5[\text{dB}]$ and $10[\text{dB}]$, (v) $\varepsilon = 0.2537$, and
(vi) $62.5\%$ of training. As a performance measure, we consider the
normalized mean square error, defined as $\text{NMSE}:=
(\textbf{h}-\hat{\textbf{h}})^H (\textbf{h}-\hat{\textbf{h}}
)/(\textbf{h}^H \textbf{h})$. Using 100 different realizations for the
noise, the results can be seen in Table \ref{table:variance_u}. We can
conclude that regularization only helps if limited amount of data is
available. For the case of $\SNR = 10[\text{dB}]$, in Fig \ref{fig:ber}, the
average BER for ML estimation is 0.0195, and for MAP estimation is
0.0132.



\section{Conclusions}

In this work, we have proposed an algorithm to estimate sparse channels in OFDM systems, the CFO, the variance of the
noise, the symbol, and the parameter defining the a priori
distribution of the sparse channel. This is achieved in
the framework of MAP estimation, using the EM algorithm.

Sparsity has been promoted by using an $\ell_1$-norm regularization,
in the form of a prior distribution for the CIR. For that, the EM
algorithm has been modified to include this case. In
addition, we have concentrated the cost function in the M-step to numerically optimize one single variable ($\varepsilon$).

The numerical examples illustrate the effectiveness of this approach
for the partial training case, obtaining, in most cases studied, a lower
value for NMSE using regularization compared to the value for NMSE
using no regularization. For the full training case, there is no
noticeable difference between the estimates obtained with ML and MAP. This confirms that prior knowledge is useful when the
amount of data is limited.

\bibliographystyle{IEEEtran}


\end{document}